\documentclass[12pt,reqno]{amsart}

\usepackage{amsfonts}
\usepackage{amscd}
\usepackage{amssymb}
\usepackage{amsmath} 
\usepackage{mathrsfs}
\usepackage{dsfont}
\usepackage{enumerate}
\usepackage{xcolor}
\usepackage{float} % To use [H]
\usepackage[pdftex]{graphicx}
\allowdisplaybreaks
\usepackage{url}
\usepackage{soul}
\usepackage{verbatim}
\usepackage[hidelinks]{hyperref} 
 \hypersetup{
     colorlinks,
     linkcolor={red!70!black},
     citecolor={blue!80!black},
     urlcolor={blue!60!black}
 }
\newtheorem*{theorem*}{Theorem}
\newtheorem{theorem}{Theorem}[section]
\newtheorem{corollary}[theorem]{Corollary}
\newtheorem{lemma}[theorem]{Lemma}
\newtheorem{proposition}[theorem]{Proposition}
\newtheorem{definition}[theorem]{Definition}

\newtheorem{assumption}[theorem]{Assumption}

\numberwithin{equation}{section}
\theoremstyle{plain}
\newtheorem*{theoremA*}{Theorem A}
\newtheorem*{theoremB*}{Theorem B}
\newtheorem*{theoremC*}{Theorem C}
\newcommand{\beas}{\begin{eqnarray*}}
\newcommand{\eeas}{\end{eqnarray*}}
\newcommand{\bes} {\begin{equation*}}
\newcommand{\ees} {\end{equation*}}
\newcommand{\be} {\begin{equation}}
\newcommand{\ee} {\end{equation}}
\newcommand{\bea} {\begin{eqnarray}}
\newcommand{\eea} {\end{eqnarray}}

\newcommand{\R}{\mathbb R}
\newcommand{\C}{\mathbb C}
\newcommand{\Z}{\mathbb Z}

\newcommand{\B}{\mathbb B}
\newcommand{\Sp}{\mathbb S}

\newcommand{\A}{\mathcal A}

\newcommand{\tr}{\operatorname{tr}} 
\numberwithin{equation}{section}

\title[]{on the Lieb--Wehrl Entropy conjecture for $SU(N,1)$}
\author[]{Mandeep Singh}	
\address{Mandeep Singh,
\endgraf Department of Mathematics,
IIT Guwahati,
\endgraf Guwahati - 781039, Assam, India.}
\email{mandeepgodara27@gmail.com, mandeep.singh@iitg.ac.in}
\date{}
\thanks{This work was supported by an Institute Fellowship of IIT Guwahati.}

\begin{document}

\begin{abstract}

We investigate the sharp functional inequalities for the coherent state transforms of $SU(N,1)$. These inequalities are rooted in Wehrl's definition of semiclassical entropy and his conjecture about its minimum value. Lieb resolved this conjecture in 1978, posing a similar question for Bloch coherent states of $SU(2)$. The $SU(2)$ conjecture was settled by Lieb and Solovej in 2014, and the conjecture was extended for a wide class of Lie groups. The generalized Lieb conjecture has been resolved for several Lie groups, including $SU(N),\, N\geq2$, $SU(1,1)$, and its $AX+B$ subgroup. Under the Li–Su assumption on isoperimetric regions in the complex hyperbolic ball, our sharp functional inequalities for the coherent state transforms extend this resolution to $SU(N,1),N\geq2$. Additionally, we explore the Faber-Krahn inequality, which applies to the short-time Fourier transform with a Gaussian window. This inequality was previously proven by Nicola and Tilli and later extended by Ramos and Tilli to the wavelet transform. In this paper, we further extend this result within the framework of the Bergman space $\A_{\alpha}$.

\end{abstract}
\maketitle

\maketitle
\section{Introduction}\label{sec:introduction}
%%%%%%%%%%%%%%%%%%%%%%%%%%%%%%%%%%%%%%%%%1111111111111111111111111111111111111111
We begin with sharp functional inequalities for coherent state transforms and their connection with Lieb -Wehrl entropy conjecture. To establish the connection, we consider the normalized Gaussian functions $\psi_z\in L^2(\R)$  parametrized by $z=(p,q)\in\R^2$ given by
\bes
\psi_z(x)= (\pi \hbar)^{-\frac{1}{4}}e^{\frac{-(x-q)^2}{2\hbar}+\frac{ipx}{\hbar}}, \text{ for all } x\in\R,
\ees
where $\hbar >0$ is a fixed constant. These vectors are (Schr\"{o}dinger, Klauder, Glauber) coherent states first considered by Schr\"{o}dinger and these are related to irreducible unitary representations of the Heisenberg-Weyl group. The analogues of these vectors for other Lie groups have been studied in the literature (see e.g. \cite{per86,per72}). For a nonnegative operator $\rho$ on $L^2(\R)$ with $\tr \rho=1$, the transform
\bes
h_\rho(z)= \langle \psi_z,\rho \psi_z\rangle,
\ees
is known as the coherent state transforms, the Husimi function, or the covariant symbol. Corresponding to $\rho$, Wehrl in \cite{weh79} defined the classical entropy by
\bes
-\int_{\R^2} h_\rho(z) \log h_\rho(z) dz,
\ees
where he proved that it is positive, and conjectured that its minimum value occurs when $\rho$ is a projection operator onto any coherent state, that is, when $\rho =|\psi_{z_0}\rangle\langle\psi_{z_0}|$ for some $z_0\in\R^2$. Lieb proved the conjecture shortly thereafter in \cite{lie78}, where he proved more generally that for $s\geq1$, the quantity
\bes
\int_{\R^2}( h_\rho(z))^sdz
\ees
is maximized when $\rho =|\psi_{z_0}\rangle\langle\psi_{z_0}|$ for some $z_0\in\R^2$. Since the above quantity is independent of $\rho$ for $s=1$, Wehrl's entropy conjecture follows by differentiating it at $s=1$. Later, Carlen gave an alternate proof of Lieb's result based on logarithmic Sobolev inequality in \cite{car91}, and he also characterized the cases of equality. Another proof by Luo is also given in \cite{luo00}. Lieb and Solovej generalized above result in \cite{lie14} and proved that for any convex function $\Phi:[0,1]\to\R$ with $\Phi(0)=0$, the integral
\bes
\int_{\R^2}\Phi( h_\rho(z))dz
\ees
is maximized when $\rho =|\psi_{z_0}\rangle\langle\psi_{z_0}|$ for some $z_0\in\R^2$. They called it the generalized Wehrl conjecture.

In \cite{lie78}, Lieb observed that the analogue of Wehrl's entropy conjecture should hold if we replace Scr\"{o}dinger coherent states by Bloch coherent states, which are related to the irreducible representations of $SU(2)$. This conjecture was resolved in \cite{lie14}, although some special cases were already known in \cite{bod04,sch99}. Lieb and Solovej used limting argument in \cite{lie14} and noted that this approach does not work for $SU(N)$, $N>2$. Later, in \cite{lie16}, they proved the conjecture for all symmetric representations of $SU(N)$. Furthermore, they conjectured that the analogue of Wehrl's conjecture should hold, at least for a wide class of Lie groups. 

Recently, in \cite{lie21}, Lieb and Solovej considered the group $SU(1,1)$ and its subgroup $AX+B$ where they formulated the Wehrl-type entropy conjecture into a problem about containment of Bergman spaces on $\mathbb D$. Kulikov proved this containment of Bergman spaces on $\mathbb D$ in \cite{kul22}, hence confirming Wehrl's conjecture for $SU(1,1)$. He also characterized the cases of equality when $\Phi$ is strictly convex. Some particular cases for $SU(1,1)$ were already considered in \cite{ban09,bay19}. Inspired by Kulikov's technique, Frank applied a uniform approach to prove Wehrl's conjecture for the Heisenberg-Weyl group, $SU(2)$, and $SU(1,1)$, and characterized the cases of equality when $\Phi$ is not affine linear (see \cite{fra23}).

We shall utilize the 
containment of Bergman spaces on $\B_N$ established by Li and Su in \cite{xia24}. 
Their argument relies on the following assumption.

\begin{assumption}\label{assump}
Isoperimetric regions in the complex hyperbolic ball 
$\B_N$ are geodesic balls.
\end{assumption}

Since this geometric assumption is not yet known in full generality, our results should be understood as conditional on it. If this assumption is eventually established, our results become unconditional. One of the main results of this paper is the following theorem, which generalizes \cite[Theorem 6]{fra23} to higher dimensions.

%%%%%%%%%%%%%%%%%%%%%%%%%%%%%%%%%%%%%%%%%%%%%%%%%%%%%%%%%%%%%%%%%%%%%%%%%%%%%%%%%%
\hypertarget{th:A}{\begin{theoremA*}[Under Assumption \ref{assump}]
   Let $k \in \Z_+$ and consider the irreducible discrete series representation of $SU(N,1)$ on the Hilbert space $\A_{(N+1)k}$ corresponding to $k$. Let $\Phi:[0,1] \rightarrow \R$ be convex. Then,
\bes
\begin{split}
    \sup \left\{ \int_{\B_N}\Phi(|\langle \phi_z,\phi \rangle|^2)dm(z) : \phi \in \A_{(N+1)k}, \ \|\phi\|_{\A_{(N+1)k}}=1\right\}\\ 
    = \frac{N}{(N+1)k}\int_0^1\Phi(s)s^{-\frac{N}{(N+1)k}-1}\left(1-s^{1/(N+1)k}\right)^{N-1}ds,
\end{split}
\ees
and the supremum is attained for $\phi = e^{i\theta}\phi_{z_0}$ for some $z_0 \in \B_N$, $\theta\in\R$. If $\Phi$ is strictly convex and if the supremum is finite, then it is attained only for such $\phi$.
\end{theoremA*}}

Note that if $\Phi(t)=t^s$ for $s\geq 1$, then the analogue of Wehrl's original entropy conjecture follows from Theorem \hyperlink{th:A}{A} by taking the negative derivative at $s=1$. One consequence of the above theorem is the following result, which generalizes \cite[Corollary 7]{fra23}.

\hypertarget{th:B}{\begin{theoremB*}[Under Assumption \ref{assump}] \label{co:generalizedwehrlconjecture}
   Let $k \in \Z_+$ and consider the irreducible discrete series representation of $SU(N,1)$ on $\A_{(N+1)k}$ corresponding to $k$. Let $\Phi:[0,1] \rightarrow \R$ be convex. Then,
\bes
\begin{split}
    \sup \left\{ \int_{\B_N}\Phi(\langle \phi_z,\rho\phi_z \rangle)dm(z) : \rho \geq 0 \text{ on }\A_{(N+1)k}, \ \tr \rho=1\right\}\\ 
    = \frac{N}{(N+1)k}\int_0^1\Phi(s)s^{-\frac{N}{(N+1)k}-1}\left(1-s^{1/(N+1)k}\right)^{N-1}ds,
\end{split}
\ees
and the supremum is attained for $\rho = |e^{i\theta}\phi_{z_0}\rangle\langle e^{i\theta}\phi_{z_0}|$ for some $z_0 \in \B_N$, $\theta\in\R$. If $\Phi$ is strictly convex and if the supremum is finite, then it is attained only for such $\rho$.
\end{theoremB*}}
%%%%%%%%%%%%%%%%%%%%%%%%%%%%%%%%%%%%%%%%%%%%%%%%%%%%%%%%%%%%%%%%%%%%%%%%%%%%%%%%%%
In this paper, we also discuss another important topic, the Faber–Krahn inequality.  Given $f\in L^2(\R)$, the Short-time Fourier transform (STFT) is defined as 
$$\nu f(x,\omega)=\int_{\R}e^{-2\pi iy\omega}f(y)\phi(x-y)dy, \;\;\;\; x,\omega \in \R$$
where $\phi(x)=2^{1/4}e^{-\pi x^2}$ is the Gaussian window. The STFT is also called \textit{Bargmann transform} in Analysis and the \textit{Coherent state transform} in quantum mechanics. We refer to \cite{nic22} and the references provided therein for more details on this topic. The Faber–Krahn inequality for STFT, as given in \cite[Theorem 1.1]{nic22}, can be conveniently restated in terms of functions in the Fock space on $\C$. In \cite{nic22}, Nicola and Tilli proved the following result for the Fock space.
\begin{theorem} \cite[Theorem 3.1]{nic22}
    For every $F\in\mathcal{F}^2(\C)\setminus\{0\}$ and every measurable set $E\subset\R^2$ of finite measure, we have
    \bes
    \frac{\int_E|F(z)|^2e^{-\pi |z|^2 }dz}{\|F\|_{\mathcal{F}^2}^2}\leq 1-e^{-|E|}.
    \ees
    Moreover, equality occurs (for some $F$ and for some $E$ such that $0<|E|<\infty$) if and only if $F=cF_{z_0}$ (for some $z_0\in\C$ and some nonzero $c\in\C$) and $E$ is equivalent, up to a set of measure zero, to a ball centered at $z_0$.
\end{theorem}
Inspired by the proof of the Faber-Krahn inequality for STFT in the work mentioned above,  Ramos and Tilli extended the Faber-Krahn inequality for Wavelet transform in \cite{ram23}. They further reduced this problem to an optimization problem on Bergman spaces. 
%%%%%%%%%%%%%%%%%%%%%%%%%%%%%%%%%%%%%%%%%%%%%%%%%%%%%%%%%%%%%%%%%%%%%%%%%%%%%%
Now, we state the second main result of this paper. The following can be considered as an extension of the above result and a version of Faber-Krahn inequality for Bergman spaces $\mathcal A_\alpha$ on $\B_N$. 

 \hypertarget{th:C}{\begin{theoremC*}[Under Assumption \ref{assump}] \label{th:faberkrahn}
    Let $\alpha >N$ and $s>0$ be fixed. Then, for every $\phi \in \A_{\alpha}$ with $\|\phi\|_{\A_{\alpha}}=1$ and every measurable set $E \subset \B_N$ such that $m(E)=s$, we have
    \be \label{eq:faberinequality}
    \int_{E}|\phi(z)|^2\left(1-|z|^2\right)^{\alpha}dm(z) \leq \int_{B_s}\left(1-|z|^2\right)^{\alpha}dm(z),
    \ee
    where $B_s$ is the ball centered at the origin such that $m(B_s)=s$. Moreover, there is equality in (\ref{eq:faberinequality}) if and only if $\phi \equiv e^{i\theta}\phi_{z_0}$ for some $z_0 \in \B_N$, $\theta\in\R$, and $E$ is equivalent (up to measure 0) to a ball centered at $z_0$, such that $m(E)=s$.
\end{theoremC*}}

The paper is organized as follows. In the next section, we introduce some preliminary notations, define Bergman spaces over $\B_N$, and list some of their properties. In Section 3, we start with the structure of the group $SU(N,1)$ , its representation, and corresponding coherent states. Then, we proceed to prove Theorem \hyperlink{th:A}{A} and Theorem \hyperlink{th:B}{B}. Finally, Section 4 includes a detailed discussion of the Faber-Krahn inequality and its extension in the Bergman space framework, where we prove Theorem \hyperlink{th:C}{C}.
%%%%%%%%%%%%%%%%%%%%%%%%%%%%%%%%%%%2222222222222222222222222222222222222222
%%%%%%%%%%%%%%%%%%%%%%%%%%%%%%22222222222222222222222222222222222222222222
\section{Notation and Preliminaries}\label{sec:notations}
 Most of our notation and results are standard and can be found in \cite{per86,rud08,zhu05}. The symbols $\mathbb{Z}_{+},\mathbb{R}$, and $\mathbb{C}$ will respectively denote the set of all nonnegative integers, real numbers, and complex numbers. Throughout the paper, we fix a positive integer $N$.
For given $z = (z_1, \dots,z_N)$ and $ w=(w_1,\dots,w_N) $ in $\C^N$,  we define
\bes
\langle z,w\rangle =z_1\overline{w}_1+\dots+z_N\overline{w}_N,
\ees
where $\overline{w}_k$ is the complex conjugate of $w_k$. We also write
\bes
| z | = \sqrt{\langle z,z \rangle} = \sqrt{|z_1|^2+\dots+|z_N|^2}.
\ees
The open unit ball in $\C^N$ is the set
\bes
\B_N = \{z\in \C^N : |z| <1\}.
\ees
The boundary of $\B_N$ will be $\Sp_N = \{z \in \C^N : |z|=1\}$, the unit sphere in $\C^N$.
%\begin{definition}
%    For $0<p<\infty$ the Hardy space $H^p$ consists of holomorphic functions $\phi$ in $\B_N$ such that
%    \bes
%    \|\phi\|_{H^p}^p = \sup_{0<r<1}\int_{\Sp_N}|\phi(rA_{u^*(s)})|^pd\sigma(A_{u^*(s)}) < \infty,
%    \ees
%  where $d\sigma$  is the surface measure on $\Sp_N$ normalized so that $\sigma(\Sp_N)=1.$
%\end{definition}
%To define the Bergman space, we 
For $z \in \B_N$, the M\"{o}bius invariant measure on $\B_N$ is given by
\bes
dm(z) = \frac{dv(z)}{(1-|z|^2)^{N+1}},
\ees
where $dv$ denotes the volume measure on $\B_N$ normalized so that $v(\B_N)=1$. Now we define the appropriate Bergman spaces.
\begin{definition}
 For every $\alpha >N$, the \textit{Bergman space} $\mathcal A_{\alpha}$ is the Hilbert space of all  holomorphic functions $\phi:\B_N\to \C$ such that
    \bes
    \|\phi\|_{\mathcal A_{\alpha}}^2= c_{\alpha}\int_{\B_N}|\phi(z)|^2(1-|z|^2)^{\alpha}dm(z) < +\infty,
    \ees
    where $c_{\alpha}$ is a normalizing constant so that $\phi(z) \equiv 1 $ has Bergman norm 1. 
\end{definition}
Using polar coordinates, the value of $c_\alpha$ is given by
    \bes
    c_{\alpha}= \frac{\Gamma(\alpha)}{N!\Gamma(\alpha-N)}.
    \ees
The inner product on these spaces is given by
\bes
\langle \phi,\psi \rangle = c_{\alpha}\int_{\B_N} \overline{\phi}(z)\psi(z)(1-|z|^2)^{\alpha}dm(z).
\ees
Given $w\in\B_N$, the reproducing kernel at $w$ is the unique function $K_w \in \A_\alpha$ such that for all $f\in \A_\alpha$, $f(w)=\langle f,K_w\rangle$. This kernel is given by 
\bes
K_w(z)=\left(1-\langle z,w\rangle\right)^{-\alpha}.
\ees
For any point $z_0 \in \B_N\setminus\{0\}$, we define biholomorphic mappings of $\B_N$ onto itself by
\bes
\Upsilon_{z_0}(z)= \frac{z_0-\frac{\langle z,z_0\rangle}{|z_0|^2}z_0-\left(1-|z_0|^2\right)^{1/2}\left(z-\frac{\langle z,z_0\rangle}{|z_0|^2}z_0\right)}{1-\langle z,z_0\rangle}.
\ees
These mappings induce surjective isometries which are also involutive (see \cite{vuk93}) given as 
\bes
\left(T_{z_0}\phi\right)(z)= \left(\frac{1-|z_0|^2}{\left(1-\langle z,z_0\rangle\right)^2}\right)^{\alpha/2}\phi\left(\Upsilon_{z_0}(z)\right).
\ees
It turns out that the functions
\bes
\phi_{z_0}(z)=\left(T_{z_0}1\right)(z)=\left(\frac{1-|z_0|^2}{\left(1-\langle z,z_0\rangle\right)^2}\right)^{\alpha/2}
\ees
have Bergman norm 1 for all admissible values of $\alpha$. Throughout the paper, for any $\phi \in \mathcal A_{\alpha}$, we set the function
\be \label{eq:u} 
u_{\phi}(z):= \left|\phi(z)\right|^2\left(1-\left|z\right|^2\right)^{\alpha},
\ee
and its superlevel sets
\be \label{eq:a_t}
\Omega_{\phi,t}:=\left\{z\in \B_N : u_{\phi}(z)>t\right\}, \text{ for } t\geq 0.
\ee
The distribution function of $u_{\phi}$ is given by
\be \label{eq:mu} 
\mu_{\phi} (t):=m(\Omega_{\phi,t}), \text{ for }0 \leq t \leq \max_{\B_N}u_\phi,
\ee
and its decreasing rearrangement is defined by
\be \label{eq:ustar}
u_{\phi}^*(s):= \sup \{t\geq 0 :\mu_{\phi}(t)>s\} \text{ for } s\geq0.
\ee
We note that $u_{\phi}^*$ is the inverse function of $\mu_{\phi}$, i.e.,
$
u_{\phi}^*(s)=\mu_{\phi}^{-1}(s), \text{ for }s \geq 0,
$
which maps $[0,+\infty)$ decreasingly and continuously onto $(0, \max u_\phi]$.\par
%It is well known fact that $H^p$ is the limit of $A_{\alpha}^{p\alpha}$ when $\alpha \rightarrow 1^+ $ in the sense that if $\phi$ is in $H^p$, then $\phi$ is in $A_{\alpha}^{p\alpha}$ for every $\alpha >1$ and
%\bes
%\lim_{\alpha \rightarrow 1^+}\|\phi\|_{A_{\alpha}^{p\alpha}} = \|\phi\|_{H^p}.
%\ees
%So, sometimes we choose the convention $A_1^p = H^p$.\\
For any $\phi \in \A_{\alpha}$, the sharp estimate for $u_\phi(z)$ in terms of $\|\phi\|_{\A_{\alpha}}$ is given in the following proposition.
\begin{proposition}\cite{vuk93} \label{prop:pointevaluation}
    For any $\phi \in \A_{\alpha}$ and for every $z\in \B_N$,
    \be \label{eq:pointevaluation}
    u_\phi(z)\leq \|\phi\|_{\A_{\alpha}}^2.
    \ee
    The equality in (\ref{eq:pointevaluation}) occurs at some point $z_0\in \B_N$ if and only if $\phi=c\phi_{z_0}$ for some $c \in \C$ with $|c|=\|\phi\|_{\A_{\alpha}}$.
    \end{proposition}
\section{Proofs of Theorem A and Theorem B}\label{sec:wehrlconjecture}
Before the proofs of Theorem \hyperlink{th:A}{A} and Theorem \hyperlink{th:B}{B}, we recall some facts about the group $SU(N,1)$ and their irreducible discrete series unitary representation (for details, we refer to \cite{per86,per72}). 
%To characterize the maximizers, we first deduce a lemma in which we discuss the equality cases of Lemma 4.1 in \cite{kal24} (see also \cite[Lemma 5.1]{xia24}). 
The group $SU(N,1)$ is the group of  $(N+1)\times(N+1)$ complex matrices of determinant 1 that leave the form $|z_1|^2+\dots+|z_{N}|^2-|z_{N+1}|^2$ invariant. If $g\in SU(N,1)$ then $g$ can be written as 
\bes
g=\begin{pmatrix}
    A&B\\C&D
\end{pmatrix}
\ees
where $A,B,C,D$ are $N\times N,N\times 1,1\times N, \text{and }1\times 1$ matrices, respectively, satisfying
\bes
AA^+-BB^+=I_N, AC^+=BD^+, DD^+-CC^+=1.
\ees
Here, $I_N$ denotes the  $N\times N$ identity matrix and $Z^+$ denotes the transposed conjugate of a matrix $Z$.
The group $SU(N,1)$ acts transitively on $\B_N$ by
\bes
z \rightarrow zg = (A'z+C')(B'z+D')^{-1},
\ees
and the coherent states of the group will be parametrized by points on $\B_N$. For a positive integer $k$, the group has irreducible discrete series unitary representation which can be realized on $\A_{(N+1)k}$ by
\bes
U^k(g)\phi(z)=\left[\det(B'z+D')\right]^{-(N+1)k}\phi(zg).
\ees
Coherent states corresponding to the lowest weight vector $\phi_0 \equiv 1$ are given by
\bes
\phi_g(z)=U^k(g)\phi_0(z)=\left[\det(B'z+D')\right]^{-(N+1)k}.
\ees
Also, $\phi_{g_1}=e^{i\theta}\phi_{g_2}$ if and only if $g_1=g_2h$, where $\theta \in \R$ and $h$ belongs to the isotropy subgroup of $\phi_0$. Therefore, as mentioned before, every coherent state is now determined by a point $w \in \B_N$, $\phi_g(z)=e^{i\theta}\phi_{w}(z)$, where
\bes
\phi_{w}(z)=\left(\frac{(1-|w|^2)}{(1-\langle z,w\rangle)^2}\right)^{(N+1)k/2}.
\ees

\begin{comment}
\begin{theorem}\label{th:generalformgeneralizedwehrlconjecture}
   Let $0<p<+\infty$ and let $\Phi:[0,1] \rightarrow \R$ be convex. Then,
\bes
\begin{split}
    \sup \left\{ \int_{\B_N}\Phi\left(|\phi(z)|^p\left(1-|z|^2\right)^{\alpha}\right)dm(z) : \phi \in A_{\alpha}^p,  \|\phi\|_{A_{\alpha}^p}=1\right\}\\ = \frac{N}{\alpha}\int_0^1\Phi(s)s^{-\frac{N}{\alpha}-1}\left(1-s^{1/\alpha}\right)^{N-1}ds,
\end{split}
\ees
and the supremum is attained for $\phi = e^{i\theta}\phi_{z_0}$ for some $z_0 \in \B_N$, and some $\theta\in[0,2\pi)$. If $\Phi$ is strictly convex and if the supremum is finite, then it is attained only for such $\phi$.
\end{theorem}
\end{comment}
\begin{comment}
\begin{corollary}\label{co:generalformgeneralizedwehrlconjecture}
Let $\Phi:[0,1] \rightarrow \R$ be convex. Then,
\bes
\begin{split}
    \sup \left\{ \int_{\B_N}\Phi(u_{\rho}(z)^2)dm(z) : \rho \geq 0 \text{ on }A_{\alpha}^2, \tr \rho=1\right\}\\ = \frac{N}{\alpha}\int_0^1\Phi(s)s^{-\frac{N}{\alpha}-1}\left(1-s^{1/\alpha}\right)^{N-1}ds,
\end{split}
\ees
and the supremum is attained for $\rho = |e^{i\theta}\phi_{z_0}\left>\right<e^{i\theta}\phi_{z_0}|$ for some $z_0 \in \B_N$, and some $\theta\in[0,2\pi)$. If $\Phi$ is strictly convex and if the supremum is finite, then it is attained only for such $\rho$.
\end{corollary}
\end{comment}
Now we are ready to prove Theorem A. 

\begin{proof}[Proof of Theorem \hyperlink{th:A}{A}]
Let $k$ be positive integer and $\phi\in \A_{(N+1)k}$. From the above discussion we have
\bes
\langle \phi_z,\phi \rangle = (1-|z|^2)^{(N+1)k/2}f(z),
\ees
where 
\bes
f(z)=c_{(N+1)k}\int_{\B_N}(1-\langle z, w \rangle)^{-(N+1)k}\phi(w)(1-|w|^2)^{(N+1)k}dm(w).
\ees
Since $(1-\langle z, w \rangle)^{-(N+1)k}$ is the reproducing kernel of $\A_{(N+1)k}$ and $\phi \in \A_{(N+1)k}$, $f(z)=\phi(z)$ for all $z\in\B_N$. Hence,
\be \label{eq:inner_product_of_phiw_related_to_uphi}
|\langle\phi_z,\phi\rangle|^2=(1-|z|^2)^{(N+1)k}|\phi(z)|^2=u_\phi(z).\ee
Now in view of \cite[Theorem 1.6]{xia24}, for $\phi \in \A_{(N+1)k}$ with  $\|\phi\|_{\A_{(N+1)k}}=1$, we have the following inequality:
\bes
    \int_{\B_N}\Phi(|\langle \phi_z,\phi \rangle|^2)dm(z) \leq  \int_{\B_N}\Phi(|\langle \phi_z, 1 \rangle|^2)dm(z),
\ees
and the supremum on the left-hand side is attained for $\phi\equiv1$. Since for every $t>0$ the measure $m(\{u_{\phi_{z_0}}(z)>t\})$ is independent of $z_0$, the supremum on the left-hand side is attained for $\phi=e^{i\theta}\phi_{z_0}$, for some $z_0\in \B_N$, $\theta \in \R$. A simple calculation by taking $\phi\equiv1$ shows that the explicit value of the supremum is 
   \bes
   \begin{split}
   \int_{\B_N}\Phi\left(\left(1-|z|^2\right)^{(N+1)k}\right)dm(z)&=2N\int_0^1\frac{\Phi\left(\left(1-r^2\right)^{(N+1)k}\right)r^{2N-1}}{\left(1-r^2\right)^{N+1}}dr \\ &= \frac{N}{(N+1)k}\int_0^1\Phi(s)s^{-\frac{N}{(N+1)k}-1}\left(1-s^{1/(N+1)k}\right)^{N-1}ds.
   \end{split}
   \ees
  To prove the second part of the theorem, we assume that the supremum is finite and the given function $\Phi$ is strictly convex. This implies that the derivative $\Phi'$ is a strictly increasing function. Now, if the supremum is attained for some $\phi$, then we have
   \be \label{eq:intmutphi't}
   \int_{\B_N}\Phi\left(|\phi(z)|^2\left(1-|z|^2\right)^{(N+1)k)}\right)dm(z)=\int_0^1\mu_{\phi}(t)\Phi'(t)dt= \int_0^1\mu_1(t)\Phi'(t)dt.
   \ee
   Also, from the norm condition $\|\phi\|_{\A_{(N+1)k}}=1$, we have
    \be \label{eq:intmut}
   \int_0^1\mu_{\phi}(t)dt= \int_0^1\mu_1(t)dt.
   \ee
   We choose $a\in[0,1]$ such that $\mu_{\phi}(t)\geq\mu_1(t)$ for $t\leq a$ and $\mu_{\phi}(t)\leq\mu_1(t)$ for $t\geq a$. Since the function $g(t)=t^{\frac{1}{(N+1)k}}(\mu^{\frac{1}{N}}_{\phi}(t)+1)$ is decreasing  (see \cite[Theorem 3.1]{xia24}), such an $a$ always exists. Hence, the function
    \be \label{eq:ht}
    h(t):=\left(\Phi'(t)-\Phi'(a)\right)\left(\mu_{\phi}(t)-\mu_1(t)\right) \leq 0 \text{ for all } t\in[0,1].
    \ee
    Now, multiplying with $-\Phi'(a)$  on both sides of the equation  (\ref{eq:intmut}) and using (\ref{eq:intmutphi't}), we obtain
    \bes
    -\Phi'(a)\int_0^1\left(\mu_{\phi}(t)-\mu_1(t)\right)dt+ \int_0^1\Phi'(t)\left(\mu_{\phi}(t)-\mu_1(t)\right)dt=0.
    \ees
    This implies,
    \bes
    \int_0^1h(t)dt=\int_0^1\left(\Phi'(t)-\Phi'(a)\right)\left(\mu_{\phi}(t)-\mu_1(t)\right)dt=0.
    \ees
    Now from (\ref{eq:ht}), we have $h(t)\equiv0$. Since $\Phi'$ is strictly-increasing, $\mu_{\phi}(t)=\mu_1(t)$ for all $t\in[0,1],$ which means $\|u_\phi\|_{\infty}=1$. Now, by appealing to Proposition \ref{prop:pointevaluation}, we have $\phi = e^{i\theta}\phi_{z_0}$ for some $z_0 \in \B_N$, $\theta\in\R$.
\end{proof}
Now we will prove Theorem B, which is a consequence of Theorem A. 
\begin{proof}[Proof of Theorem \hyperlink{th:B}{B}]
For a nonnegative operator $\rho$ on $\A_{(N+1)k}$ with $\tr \rho =1$, we have 
\bes
\rho=\sum_{i}\lambda_i|\psi_i\rangle\langle \psi_i| \text{ with } \sum_i\lambda_i=1, ~\lambda_i\geq 0,~\langle \psi_i, \psi_j\rangle = \delta_{i,j}.
\ees
Therefore, 
\bes
\langle \phi_z,\rho \phi_z\rangle=\langle \phi_z,\left(\sum_{i}\lambda_i|\psi_i\rangle\langle \psi_i|\right) \phi_z\rangle=\left< \phi_z,\sum_{i}\lambda_i\langle \psi_i,\phi_z\rangle \psi_i \right>=\sum_{i}\lambda_i|\langle \psi_i,\phi_z\rangle|^2.
\ees
Next, by using (\ref{eq:inner_product_of_phiw_related_to_uphi}) we have
\bes
\langle \phi_z,\rho \phi_z\rangle=\sum_{i}\lambda_iu_{\psi_i}(z).
\ees
\begin{comment}
\bes
u_{\rho}(z):=\left(\sum_i\lambda_iu_{\psi_i}(z)\right)^{\frac{1}{2}} \text{ for } z\in \B_N.
\ees
\end{comment}
Since $\Phi$ is a convex function, for any $z\in\B_N$, we have 
    \be \label{eq:Phiurho2z}
    \Phi(\langle \phi_z,\rho \phi_z\rangle)=\Phi\left(\sum_i\lambda_iu_{\psi_i}(z)\right)\leq \sum_i\lambda_i\Phi\left(u_{\psi_i}(z)\right).
    \ee
    If we denote by $S$ (which may also be $\infty$), the supremum in Theorem A, then $\int_{\B_N}\Phi\left(u_{\psi_i}(z)\right)dm(z)\leq S$ and 
    \be \label{eq:intPhiurho2}
    \int_{\B_N}\Phi(\langle \phi_z,\rho \phi_z\rangle)dm(z)\leq \sum_i\lambda_i\int_{\B_N}\Phi\left(u_{\psi_i}(z)\right)dm(z)\leq\sum_i\lambda_iS=S.
    \ee
Note that if $\rho = |e^{i\theta}\phi_{z_0}\left>\right<e^{i\theta}\phi_{z_0}|$, then 
$\langle \phi_z,\rho \phi_z\rangle= |\langle \phi_{z_0}, \phi_z\rangle|^2=u_{\phi_{z_0}}(z)$. Now, by Theorem \hyperlink{th:A}{A} the supremum will be achieved for $\rho = |e^{i\theta}\phi_{z_0}\left>\right<e^{i\theta}\phi_{z_0}|$ for some $z_0 \in \B_N$, $\theta\in \R$.

   Now, we discuss the equality case. If $S<\infty$ and it is attained for some $\rho$, then there is equality everywhere in (\ref{eq:Phiurho2z}) and (\ref{eq:intPhiurho2}). Furthermore, if $\Phi$ is strictly convex, then from the last inequality in (\ref{eq:Phiurho2z}), for every $i$ we have $u_{\psi_i}\equiv u_{\psi_1}$  and hence, $\psi_i(z)=\psi_1(z)$ for every $i$ and $z\in\B_N$. Since $\langle\psi_i,\psi_1\rangle=\delta_{i,1}$, $\psi_i\equiv 0$ for all $i\neq1$. We conclude that $\psi_1=e^{i\theta}\phi_{z_0}$ for some $z_0 \in \B_N$, $\theta\in\R$ and hence, $\rho = |e^{i\theta}\phi_{z_0}\left>\right<e^{i\theta}\phi_{z_0}|$.
\end{proof}
%%%%%%%%%%%%%%%%%%%%%%%%%%%%%%%%%%%%%%%%%%44444444444444444444444444444444444444444
%%%%%%%%%%%%%%%%%%%%%%%%%%%%44444444444444444444444444444444444444444444444444444444
\section{Faber-Krahn type inequality}\label{sec:faberkrahn}
\noindent In this section, we establish the Faber-Krahn inequality and prove Theorem \hyperlink{th:C}{C} in the setting of Hilbert spaces, namely $\A_\alpha$. We begin with some preliminary lemmas that will be needed in the proof of Theorem \hyperlink{th:C}{C}. We state our first lemma, which is an appropriate adaptation of the proof of Lemma 3.2 in \cite{nic22}. This lemma links the derivative of the distribution function and the derivative of the rearrangement function with the integral over the boundary of the superlevel sets. The expression of a derivative of the distribution function can also be found in the proof of Theorem 3.1 in \cite{xia24}.
\begin{lemma}
    The function $\mu_{\phi}(t)$ is absolutely continuous on $(0,\max u_\phi]$, and
    \be \label{eq:mu'}
    -\mu_{\phi}'(t)=\int_{u_\phi=t}\left|\widetilde{\nabla}u_\phi\right|_g^{-1}d\sigma_g.
    \ee
    Similarly, the function $u_{\phi}^*$ is absolutely continuous on $[0,+\infty)$ with
    \be \label{eq:u*'}
    -(u_\phi^*)'(s)= \left(\int_{u_\phi=u_\phi^*(s)}\left|\widetilde{\nabla}u_\phi\right|_g^{-1}d\sigma_g\right)^{-1}.
    \ee
\end{lemma}
 As Kulikov pointed out in \cite{kul22}, the above lemma has an interesting geometric interpretation in the sense that when $t$ increases by a small number $\epsilon$, the superlevel set $\Omega_{\phi,t-\epsilon}$ expands in the direction orthogonal to $\partial\Omega_{\phi,t-\epsilon}$ by the value proportional to $\epsilon\big/|\widetilde{\nabla}u_\phi|_g$. The next result is the following lemma which gives the expression for the derivative of the integral function of $u_\phi$ over its superlevel sets and follows as Lemma 3.4 in \cite{nic22} by obvious modifications. 
\begin{lemma}
    The function 
    \be \label{eq:i}
    I_\phi(s) =  \int_{\Omega_{\phi,u_\phi^*(s)}}u_\phi(z)dm(z),\text{ for } s \in [0,+\infty),
    \ee
    is of class $C^1$ on $[0,+\infty)$, and
    \be \label{eq:i'}
    I_\phi'(s)=u_\phi^*(s)  \text{ for }  s\geq 0.
    \ee
\end{lemma} 
Let $B_s$ denote a ball centered at the origin with $m(B_s)=s$ and hyperbolic radius $\rho$. Specifically,
\bes
B_s=\{z\in\B_N:|z|<\tanh \rho\}.
\ees
Then, we have the following relation between $\rho$ and $s$ which shows that the measure of a ball has exponential growth in the hyperbolic radius.
\be\label{hyperbolic measure}
        \begin{split}
             s=\int_{|w|<\tanh \rho}dm(w)
             =&\int_{|w|<\tanh \rho}\frac{1}{\left(1-|w|^2\right)^{N+1}}dv(w)\\
             =&2N\int_{0}^{\tanh \rho}\frac{\kappa ^{2N-1}}{\left(1-\kappa^2\right)^{N+1}}d\kappa\\
                          =&\left(\sinh \rho \right)^{2N}.
        \end{split}
\ee
 We shall utilize the above relation between $s$ and $\rho$ in the following lemma, where we establish the expression for the first and second order derivatives of the integral function of $(1-|z|^2)^{\alpha}$ over the balls of measure $s$. The explicit expression of this integral function can be computed directly for $N=1$, and the required derivatives can be obtained with the help of that expression. However, the situation is different for $N>1$, and we have included a detailed proof of the lemma using an approach similar to that used in the proof of Theorem 2.1 from \cite{kul22}.
    \begin{lemma}[Under Assumption \ref{assump}]
        Let $\alpha > N$ be fixed and $B_s$ %=\{z\in \B_N : |z|<r\}% 
        be a ball centered at the origin such that $m(B_s)=s$. Define the function $J(s)$ as
        \be\label{eq:j}
        J(s)=\int_{B_s}\left(1-|z|^2\right)^{\alpha}dm(z).
        \ee
       Then, we have
        \be\label{eq:j'}
        J'(s)=\left(1+ s^{1/N}\right)^{-\alpha},
        \ee
        and
        \be\label{eq:j''}
        J''(s)=-\frac{\alpha sJ'(s)}{Ns^{\frac{2N-1}{N}}+Ns^2}.
        \ee
    \end{lemma}
    \begin{proof}
        If $\rho$ is the hyperbolic radius of $B_s$, then by relation (\ref{hyperbolic measure}) between $s$ and $\rho$, we can express the elements of $B_s$ as follows:
       \bes
       \begin{split}
           B_s %=\left\{z\in \B_N : |z|<\tanh \rho\right\}
           %= \left\{z\in \B_N : \left(1-|z|^2\right)^{\alpha}>\frac{1}{\cosh ^{2\alpha}\rho} \right\}\\
            = \left\{z\in \B_N : \left(1-|z|^2\right)^{\alpha}>\frac{1}{\left(1+\sinh ^{2}\rho\right)^{\alpha}} \right\}
            = \left\{z\in \B_N : \left(1-|z|^2\right)^{\alpha}>\frac{1}{\left(1+s^{1/N}\right)^{\alpha}} \right\}.
       \end{split}
       \ees
     For $u_1(z)=\left(1-|z|^2\right)^{\alpha}$, we consider $h(z)=\chi_{B_s}(z)u_1(z)|\widetilde{\nabla}u_1|_g^{-1}$ and apply the Coarea formula as given in \cite{xia24} to write
       \bes
       \begin{split}
            J(s)=\int_{\B_N}\chi _{B_s}(z)u_1(z)dm(z)
            =\int_{\left(1+s^{1/N}\right)^{-\alpha}}^1\left(\int_{u_1(z)=\kappa}u_1(z)|\widetilde{\nabla}u_1|_g^{-1}d\sigma_g\right)dk,
       \end{split}
      \ees
      where $d\sigma_g$ denotes the hyperbolic surface area measure on $\partial B_s=\left\{ z\in\B_N:u_1(z)=\kappa\right\}$ induced by the Bergman metric on $B_N$. We deduce that
      \be \label{eq:j'begining}
      J'(s)=\frac{\alpha}{N}\left(1+s^{1/N}\right)^{-2\alpha-1}s^{\frac{1}{N}-1}\int_{u_1(z)=\left(1+s^{1/N}\right)^{-\alpha}}|\widetilde{\nabla}u_1|_g^{-1}d\sigma_g.
      \ee
      Next, we apply Cauchy-Schwartz inequality which is indeed an equality in the case of the hyperbolic surface area of $\partial B_s$ to get
      \bes
      \begin{split}
      \left(\int_{\partial B_s}|\widetilde{\nabla}u_1|_g^{-1}d\sigma_g\right) =\frac{\left(\int_{\partial B_s}d\sigma_g\right)^2}{\left(\int_{\partial B_s}|\widetilde{\nabla}u_1|_gd\sigma_g\right)}. 
       \end{split}
      \ees
      Using the isoperimetric inequality on $\B_N$ (see \cite{xia24}), the numerator term in the above expression is $4N^2s^{\frac{2N-1}{N}}+4N^2s^2$. Now we shall compute the denominator term. Let $\nu$ be the outward unit normal to $\partial B_s$ with respect to the Bergman metric. Then, we have $|\widetilde{\nabla}u_1|_g=-\langle\widetilde{\nabla}u_1,\nu\rangle_g$. Since for $z\in \partial B_s$ we have $u_1(z)=\left(1+s^{1/N}\right)^{-\alpha}$, we obtain
      \bes
      \frac{|\widetilde{\nabla}u_1|_g}{\left(1+s^{1/N}\right)^{-\alpha}}= \frac{|\widetilde{\nabla}u_1|_g}{u_1}= -\frac{\langle\widetilde{\nabla}u_1,\nu\rangle_g}{u_1}=-\langle\widetilde{\nabla}\log u_1,\nu\rangle_g.
      \ees
      Using the Gauss divergence theorem, we have
      \bes
      \begin{split}
      \int_{\partial B_s}|\widetilde{\nabla}u_1|_gd\sigma_g=&\left(1+s^{1/N}\right)^{-\alpha}\left(\int_{\partial B_s}\frac{|\widetilde{\nabla}u_1|_g}{u_1}d\sigma_g\right)\\=&-\left(1+s^{1/N}\right)^{-\alpha}\left(\int_{\partial B_s}\langle\widetilde{\nabla}\log u_1,\nu\rangle_gd\sigma_g\right)\\=&-\left(1+s^{1/N}\right)^{-\alpha}\left(\int_{B_s}\widetilde{\Delta}\left(\log u_1(z)\right) dm(z)\right).
       \end{split}
      \ees
      Note that $u_1(z) \neq 0$ for $z \in B_s$. As a result, $\log u_1(z)$ is well defined on $B_s$ and $\partial B_s$, and a simple computation shows that $\widetilde{\Delta}\left(\log u_1\right)=-4N\alpha$. Therefore, we obtain
      \bes
      \int_{\partial B_s}|\widetilde{\nabla}u_1|_gd\sigma_g=4N\alpha s\left(1+s^{1/N}\right)^{-\alpha}.
      \ees
      From the above calculation, we deduce from equation (\ref{eq:j'begining}) that
      \bes
      J'(s)=\left(1+ s^{1/N}\right)^{-\alpha}.
      \ees
      After differentiating the expression of $J'(s)$ w.r.t. $s$, we arrive at the desired equation (\ref{eq:j''}).
    \end{proof}
    Now we prove the last lemma of this section that is required to establish the equality cases in Theorem C.
    \begin{lemma}[Under Assumption \ref{assump}]\label{lem:propertyofg}
        Let $\phi \in \A_{\alpha}$ with $\|\phi\|_{\A_{\alpha}}=1$ and consider $I_\phi(s)$ and $J(s)$ as defined in (\ref{eq:i}) and (\ref{eq:j}) respectively. Then, for the function $G(s):=I_\phi(s)-J(s)$, the following statements are equivalent
        \begin{enumerate}
        \item $G'(0)=0$.
            \item $\phi \equiv e^{i\theta}\phi_{z_0}$ for some $z_0 \in \B_N$, $\theta\in \R$.
            \item $G(s)=0$ for all $s\geq 0$.
        \end{enumerate}
    \end{lemma}
    \begin{proof}
        ($1\Longleftrightarrow 2$) We have $G'(0)=I_\phi'(0)-J'(0)=u_\phi^*(0)-1=\|u_\phi\|_{\infty}-1$. Therefore, $G'(0)=0$ if and only if $\|u_\phi\|_{\infty}=1$. It follows from the equality part of Proposition \ref{prop:pointevaluation} that $\|u_\phi\|_{\infty}=1$ if and only if $\phi \equiv e^{i\theta}\phi_{z_0}$ for some $z_0 \in \B_N$, $\theta\in \R$.

\smallskip

        ($2 \Longrightarrow 3$) If we consider $\phi \equiv e^{i\theta}\phi_{z_0}$ for some $z_0 \in \B_N$, $\theta\in \R$, then
        \bes
        u_\phi(z)=\left|e^{i\theta}\phi_{z_0}(z)\right|^2\left(1-|z|^2\right)^{\alpha}=\left(\frac{\left(1-|z_0|^2\right)\left(1-|z|^2\right)}{|1-\langle z,z_0\rangle|^2}\right)^{\alpha}=\left(1-|\Upsilon_{z_0}(z)|^2\right)^{\alpha}.
        \ees
        Therefore, $\Omega_{\phi,u_\phi^*(s)}$ is a ball centered at $z_0$ and $m(\Omega_{\phi,u_\phi^*(s)})=\mu_{\phi} (u_\phi^*(s))=s$. Using the change of variable as given in \cite[Equation (1.26)]{zhu05}, we have, for every $s\geq0$, 
        \bes
        I_\phi(s)= \int_{\Omega_{\phi,u_\phi^*(s)}}\left(1-|\Upsilon_{z_0}(z)|^2\right)^{\alpha}dm(z)=\int_{B_s}\left(1-|z|^2\right)^{\alpha}dm(z)=J(s).
        \ees
        ($3 \Longrightarrow1$) This is obvious. With that we conclude the proof of the lemma.
    \end{proof}
    Now, we are ready to prove the main result of this section.
\begin{proof}[Proof of Theorem \hyperlink{th:C}{C}]
We will begin by proving the inequality part of the theorem. Let $\phi\in\A_\alpha$ be a function with Bergman norm 1, and let $E\subset \B_N$ be a set with $m(E)=s$. Then we can write
    $$E=(E\cap \Omega_{\phi,u_\phi^*(s)})\cup(E\backslash \Omega_{\phi,u_\phi^*(s)})\;\;\; \text{and}\;\;\ \Omega_{\phi,u_\phi^*(s)}=(\Omega_{\phi,u_\phi^*(s)}\cap E )\cup(\Omega_{\phi,u_\phi^*(s)}\backslash E)$$ 
    where, $\Omega_{\phi,u_\phi^*(s)}=\{u_\phi>u_\phi^*(s)\}$.
    Observe that $$u_\phi>u_\phi^*(s)\;\;\;  \text{on}\;\; \Omega_{\phi,u_\phi^*(s)}\backslash E\;\;\;\text{and}\;\; u_\phi\leq u_\phi^*(s)\;\;\ \text{on} \;\; E\backslash \Omega_{\phi,u_\phi^*(s)}.$$ 
    In addition, $$m(E)=m(\Omega_{\phi,u_\phi^*(s)})=s\;\; \text{implies}\;\;\  m(E\backslash \Omega_{\phi,u_\phi^*(s)})=m(\Omega_{\phi,u_\phi^*(s)}\backslash E).$$ 
    Therefore, we have the following reduction
    \be \label{eq:integrallessthani}
    \int_{E}u_\phi(z)dm(z) \leq I_\phi(s)=\int_{\Omega_{\phi,u_\phi^*(s)}}u_\phi(z)dm(z).
    \ee
   To  prove the the inequality (\ref{eq:faberinequality}), it is enough to show that $G(s)\leq0$ for $s>0,$ where $G(s)$ is defined in the previous lemma. From equations (\ref{eq:u*'}) and (\ref{eq:i'}), we have
   \bes
   I_\phi''(s)=-\left(\int_{u_\phi=u_\phi^*(s)}\left|\widetilde{\nabla}u_\phi\right|_g^{-1}d\sigma_g\right)^{-1}.
   \ees
   By Cauchy Schwartz inequality, we have
   \be  \label{eq:i'' beginning}
   I_\phi''(s) \geq -\frac{\left(\int_{u_\phi=u_\phi^*(s)}|\widetilde{\nabla}u_\phi|_gd\sigma_g\right)}{\left(\int_{u_\phi=u_\phi^*(s)}d\sigma_g\right)^2}.
   \ee
    Using the isoperimetric inequality on $\B_N$ (see \cite{xia24}), the denominator term in the above expression is bounded below by $4N^2s^{\frac{2N-1}{N}}+4N^2s^2$. To estimate the numerator term, let $\nu$ be the outward unit normal to $\partial\Omega_{\phi,u_\phi^*(s)}$ with respect to the Bergman metric. Then, we have $|\widetilde{\nabla}u_\phi|_g=-\langle\widetilde{\nabla}u_\phi,\nu\rangle_g$ and for $z\in \partial \Omega_{\phi,u_\phi^*(s)}$, we have $u_\phi(z)=u_\phi^*(s)$. Therefore,
      \bes
      \frac{|\widetilde{\nabla}u_\phi|_g}{u_\phi^*(s)}= \frac{|\widetilde{\nabla}u_\phi|_g}{u_\phi}= -\frac{\langle\widetilde{\nabla}u_\phi,\nu\rangle_g}{u_\phi}=-\langle\widetilde{\nabla}\log u_\phi,\nu\rangle_g.
      \ees
      Applying Gauss divergence theorem, we obtain
      \be \label{eq:gradiantuevaluation}
      \int_{\partial \Omega_{\phi,u_\phi^*(s)}}|\widetilde{\nabla}u_\phi|_gd\sigma_g=-u_\phi^*(s)\left(\int_{\partial \Omega_{\phi,u_\phi^*(s)}}\langle\widetilde{\nabla}\log u_\phi,\nu\rangle_gd\sigma_g\right)=-u_\phi^*(s)\left(\int_{\Omega_{\phi,u_\phi^*(s)}}\widetilde{\Delta}\left(\log u_\phi(z)\right) dm(z)\right).
      \ee
      As $u_\phi(z) \neq 0$ for $z \in \Omega_{\phi,u_\phi^*(s)}$, $\log u_\phi(z)$ is well defined on $\Omega_{\phi,u_\phi^*(s)}$ as well on $\partial \Omega_{\phi,u_\phi^*(s)}$. We have $\widetilde{\Delta}\log u_\phi(z)=2\widetilde{\Delta}\log|\phi(z)|+\alpha \widetilde{\Delta}\log\left(1-|z|^2\right)$. As $\phi(z) \neq 0$ for $z \in \Omega_{\phi,u_\phi^*(s)}$, the first term $ 2\widetilde{\Delta}\log|\phi(z)|=0$ while the second is $\alpha \widetilde{\Delta}\log\left(1-|z|^2\right)=-4N\alpha$. Thus, the right-hand side of (\ref{eq:gradiantuevaluation}) is equal to $4N\alpha su^*(s)$.
      Finally, we have
      \be \label{eq:i''estimate}
       I_\phi''(s)\geq-\frac{\alpha sI_\phi'(s)}{Ns^{\frac{2N-1}{N}}+Ns^2}.
       \ee
       %Now, we use the same line of argument \cite[Theorem 3.1]{ram23} to complete the proof. 
       If we take $h(s)=\left(1+s^{1/N}\right)^{\alpha}$, then it turns out that $(hG')'(s) \geq 0$ for $s\geq0$, that is, $hG'$ is an increasing function. Note that,
       \bes
       I_\phi(0)=J(0)=0 \quad \text{and} \quad \lim_{s\rightarrow +\infty}I_\phi(s)= \lim_{s\rightarrow +\infty}J(s)=1.
       \ees
       This implies,
       \bes
       G(0)=0 \quad \text{and} \quad \lim_{s\rightarrow +\infty}G(s)=0.
       \ees
       It follows from (\ref{eq:pointevaluation}) that $I_\phi'(0)=\|u_\phi\|_{\infty}\leq \|\phi\|_{A_{\alpha}}^2=1$, which further implies
       \bes
       G'(0)=I_\phi'(0)-J'(0)\leq0.
       \ees
       From Lemma \ref{lem:propertyofg}, $G'(0)=0$ if and only if $G(s)=0$ for all $s\geq0$. The other possibility is the case when $G'(0)<0$. In this case we show that $G(s)<0$ for all $s>0$. If possible, let $G(r_1)\geq0$ for some $r_1>0$. If we set
       \bes
       r_0:=\inf\{r>0:G(r)\geq0\},
       \ees
       then $G(r_0)=0$. Since  $G(0)=0$, by Rolle's theorem there exists $s_0\in (0,r_0)$ such that $G'(s_0)=0$. 
              Clearly $G(s_0)<0$ and hence, $G(r_0)>G(s_0)$. Therefore, we can find an $s_1 \in (s_0,r_0)$ such that $G'(s_1)>0$. As $G(r_0)=0$ and $\lim_{s\rightarrow +\infty}G(s)=0$, by Rolle's theorem there will be an $s_2 \in (r_0,+\infty)$ such that $G'(s_2)=0$. Therefore, we have  $(hG')(s_0)=0=(hG')(s_2)$, and $(hG')(s_1)>0$. This leads to a contradiction as 
       $hG'$ is an increasing function. This proves that $G(s)\leq0$ for $s>0,$ and the inequalty (\ref{eq:faberinequality}).

       Now, we consider the equality case. From the above discussion, we note that either $G(s)<0$ for all $s>0$ or $G(s)=0$ for all $s>0$. Suppose there is equality in (\ref{eq:faberinequality}) for some $s_0>0$. Then $G(s_0)=0$, and this means $G(s)=0$ for all $s>0$. Hence, from Lemma \ref{lem:propertyofg} it follows that $\phi\equiv e^{i\theta}\phi_{z_0}$ for some $z_0 \in \B_N$, $\theta\in \R$. Also, $E$ must coincide (up to measure 0) with $\Omega_{\phi,u_\phi^*(s)}$ (otherwise we would have strict inequality in (\ref{eq:integrallessthani})).

       Conversely, suppose $\phi\equiv e^{i\theta}\phi_{z_0}$ for some $z_0 \in \B_N$, $\theta\in \R$. Then, it follows from Lemma \ref{lem:propertyofg} that $G(s)=0$ for all $s\geq0$ and $\{u_\phi>u_\phi^*(s)\}$ is a ball centered at $z_0$. Also, since $E$ is equivalent (up to measure 0) to a ball centered at $z_0$ , with $m(E)=s$, there is equality in (\ref{eq:integrallessthani}). As a result, there is equality in (\ref{eq:faberinequality}).
    \end{proof}
{\bf Conflict of Interest}: The author have no conflicts of interest to declare that are relevant to the content of this article.

{\bf Data Availability}: No data sets were generated or analysed during the study. We confirm that all the data are included in this article.
\section*{Acknowledgments}
I would like to express my sincere gratitude to Pratyoosh Kumar for his constant help, insightful advice, and encouragement throughout the course of this work.

\bibliography{Ref}
\bibliographystyle{plainurl}
\end{document}